\documentclass[]{article}

\usepackage{amsopn}
\usepackage{amsmath}
\usepackage{amssymb}
\usepackage{hyperref}
\usepackage{amsthm}
\usepackage{graphicx}
\usepackage{color}
\usepackage{listings}
\usepackage{subfig}
\usepackage{enumerate}
\usepackage{authblk}
\usepackage[ruled,vlined]{algorithm2e}
\usepackage{todonotes}

\usepackage{siunitx}
\usepackage{makecell}
\usepackage{booktabs}
\usepackage{array}
\newcolumntype{x}[1]{>{\centering\arraybackslash}p{#1}}
\usepackage{tikz}
\newcolumntype{C}[1]{>{\centering\arraybackslash}p{#1}}

\newcommand\diag[4]{%
	\multicolumn{1}{p{#2}|}{\hskip-\tabcolsep
		$\vcenter{\begin{tikzpicture}[baseline=0,anchor=south west,inner sep=#1]
			\path[use as bounding box] (0,0) rectangle 
			(#2+2\tabcolsep,\baselineskip);
			\node[minimum width={#2+2\tabcolsep-\pgflinewidth},
			minimum  height=\baselineskip+\extrarowheight-\pgflinewidth] (box) 
			{};
			\draw[line cap=round] (box.north west) -- (box.south east);
			\node[anchor=south west] at (box.south west) {#3};
			\node[anchor=north east] at (box.north east) {#4};
			\end{tikzpicture}}$\hskip-\tabcolsep}}

\newcommand{\ie}{{\emph{i.e.\/}}}
\newcommand{\etal}{\emph{et al.}}

\newtheorem{definition}{Definition}

\newcommand{\C}{\ensuremath{\mathbb{C}}}

\newcommand{\ket}[1]{\ensuremath{|#1\rangle}}
\newcommand{\bra}[1]{\ensuremath{\langle#1|}}
\newcommand{\ketbra}[2]{\ensuremath{\ket{#1}\bra{#2}}}
\newcommand{\proj}[1]{\ensuremath{\ketbra{#1}{#1}}}
\newcommand{\braket}[2]{\ensuremath{\langle{#1}|{#2}\rangle}}

\newcommand{\1}{{\rm 1\hspace{-0.9mm}l}}

\newcommand{\connected}{\sim}
\newcommand{\SPAN}{\mathrm{span}}

\newcommand{\Urm}{\ensuremath{\mathrm{U}}}

\newcommand{\HH}{\mathcal{H}}
\newcommand{\XX}{\mathcal{X}}

\newcommand{\NOT}{\sigma_x}
\newcommand{\idop}[1][\XX]{\ensuremath{\1_{#1}}}

\newtheorem{theorem}{Theorem}

\newtheorem{example}{Example}

\usepackage{tikz}
\usetikzlibrary{topaths,calc}
\def\hyperpic{
\begin{tikzpicture}
\node (v0) at (0,0) {};
\node (v1) at (2,0) {};
\node (v2) at (4,0) {};
\node (v3) at (1,1.732) {};
\node (v4) at (3,1.732) {};
\node (v5) at (1.1,3.464) {};
\node (v6) at (3.1,3.464) {};
\node (m12) at (1,0) {};
\node (m24) at (1.5,0.866) {};
\node (m41) at (0.5,0.866) {};
\node (m23) at (3,0) {}; 
\node (m35) at (3.5,0.866) {}; 
\node (m52) at (2.5,0.866) {}; 
\node (m45) at (2,1.732) {}; 
\node (m56) at (2.5,2.598) {}; 
\node (m64) at (1.5,2.598) {};
\begin{scope}[fill opacity=0.7]
\filldraw[fill=yellow!70] ($(v0)+(-0.25,-0.25)$) 
    to[out=-60,in=180] ($(m12)$) 
    to[out=0,in=240] ($(v1) + (0.25,-0.25)$) 
    to[out=60,in=-60] ($(m24)$) 
    to[out=120,in=0] ($(v3) + (0,0.3535)$)
    to[out=180,in=60] ($(m41)$) 
    to[out=240,in=120] ($(v0)+(-0.25,-0.25)$);
\filldraw[fill=blue!50, thick, dotted] ($(v1)+(-0.25,-0.25)$) 
    to[out=-60,in=180] ($(m23)$) 
    to[out=0,in=240] ($(v2) + (0.25,-0.25)$) 
    to[out=60,in=-60] ($(m35)$) 
    to[out=120,in=0] ($(v4) + (0,0.3535)$)
    to[out=180,in=60] ($(m52)$) 
    to[out=240,in=120] ($(v1)+(-0.25,-0.25)$);
\filldraw[fill=red!50, dashed] ($(v3)+(-0.4,-0.4)$) --
    ($(v4) + (0.4,-0.4)$) -- 
    ($(v6) + (0.4,0.4)$) --
    ($(v5) + (-0.4,0.4)$) --
    ($(v3)+(-0.4,-0.4)$);
\end{scope}
\foreach \v in {0,1,...,6} {\fill (v\v) circle (0.06) node [below] {$v_{\v}$};}
\node at (1,0.5) {$e_0$};
\node at (3,0.5) {$e_1$};
\node at (2,2.232) {$e_2$};
\end{tikzpicture}
}

\begin{document}

\title{Quantum walks on hypergraphs}

\author[1]{Przemys{\l}aw Sadowski}

\author[1]{{\L}ukasz Pawela}

\author[1,2]{Paulina Lewandowska\footnote{Corresponding author, E-mail:
		plewandowska@iitis.pl}}

\author[1,2]{Ryszard Kukulski}

\affil[1]{Institute of Theoretical and Applied Informatics, Polish Academy
	of Sciences, ulica Ba{\l}tycka 5, 44-100 Gliwice, Poland}
\affil[2]{Institute of Mathematics, University of Silesia, ul. Bankowa 14,
	40-007 Katowice, Poland}
\date{\empty}

\maketitle

\begin{abstract}
In this work we introduce the concept of a quantum walk on a hypergraph. We 
show that the staggered quantum walk model is a special case of a quantum walk 
on a hypergraph.
\end{abstract}
\section{Introduction}
\setlength{\extrarowheight}{0.1cm}
Quantum walks may be seen as an extension of the classical random walks into the
quantum realm. There is, however, one key difference. In the classical setting
the \emph{randomness} is built-into the process. In the quantum case, the entire
process is unitary, hence deterministic and even reversible. The randomness
comes only from the random nature of quantum measurements. 

During the last two decades the field of quantum walks has received a lot of 
attention from the scientific community. One of the earliest studies are the 
works by Aharonov~\cite{aharonov2001quantum} and Kempe~\cite{kempe2003quantum}. 
Soon afterwards the possibility for algorithmic applications was 
shown~\cite{ambainis2003quantum}. One notable application is the fact that 
Grover's search algorithm~\cite{grover1996fast} can be represented as a quantum 
walk. Another approach to database lookup is the quantum spatial search 
algorithm~\cite{childs2004spatial}. Finally, nontrivial results in the field of quantum games can be obtained even with a simple walk on a cycle~\cite{miszczak2014quantum}, and some more exotic problems like the Parrondo paradox can be modeled as a quantum 
walk~\cite{pawela2013cooperative}.

Since these seminal works a lot of different approaches to the concept of a 
quantum walk have emerged. We should note here the open quantum walk 
model~\cite{attal2012open,attal2012open2}. This model can be summarized as 
follows. Imagine we have a particle moving on a graph. The particle has a 
quantum state associated with it. With each transition from one vertex to  
another, the state is modified according to some quantum operation. The only 
restriction here is that all the operations associated with some vertex must 
sum to a proper quantum channel. There was a lot effort put into 
studying this approach. We should mention here various asymptotic results for 
this model~\cite{attal2015central,sadowski2016central}, hitting times 
studies~\cite{pawela2015generalized} and potential 
applications of this model in quantum modeling of biological 
structures~\cite{chia2016coherent}

Another model which deserves mention is the quantum stochastic
walk~\cite{whitfield2010quantum}. This approach is based on the
Gorini-Kossakowski-Sudarshan-Lindblad~\cite{gorini1976completely,lindblad1976generators}
 master equation. It allows to smoothly interpolate between classical and 
quantum walks as well as gives raise to some new dynamics. The asymptotic 
behavior of this model has been extensively 
studied~\cite{domino2016properties,domino2017superdiffusive}.

Finally, there has been a lot of effort put into the extension of the standard 
unitary quantum walk. Let us note here the Szegedy walk 
model~\cite{szegedy2004quantum} which allows for quantization of arbitrary 
Markov chain based algorithms. One of the most prominent example of usage of 
this model is the quantum Page Rank 
algorithm~\cite{paparo2012google}. Another example of such modification is the 
staggered walk model introduced by Portugal 
\etal~\cite{portugal2016staggered,konno2018partition,portugal2017staggered}. 
It has applications in quantum search algorithms~\cite{tulsi2016robust}.

In this work we introduce a novel concept - quantum walks on hypergraphs. Our 
main motivation is presented in Table~\ref{tab:motiv}. In there, we show how 
the currently developed quantum walk models are constructed. The goal of this 
work is to fill the part represented by the question mark.

\begin{table}
\centering
\begin{tabular}{x{2cm}|x{2cm}x{2cm}}
	\multicolumn{3}{c}{ \large \textbf{ Quantum walks summary}} 
	\tabularnewline[2ex]
	\diag{.1em}{2cm}{Range}{Unitary}&Reflection & Arbitrary\\ \hline
	Edges&Szegedy & coined\\
	Cliques&staggered & ?\\
\end{tabular}
\caption{Summary of existing quantum walk models. The aim of this work is to
find the model which fills the gap denoted by the question mark. By ``range'' we
mean where the unitaries involved in the model act. }\label{tab:motiv}
\end{table}

This work is organized as follows. In Section~\ref{sec:graphs} we introduce the 
concept of a hypergraph along with some accompanying definitions. Next, in 
Section~\ref{sec:walks} we recall well-established quantum walk models. 
Section~\ref{sec:model} introduces our model--quantum walk on a hypergraph, or 
hyperwalk. Next, in Section~\ref{sec:relation} how our model relates to other 
quantum walks. Finally, in Section~\ref{sec:conc} final conclusions are drawn.

\section{Graphs and hypergraphs}\label{sec:graphs}
In this section we provide basic definitions used throughout this work. We 
start with the definition of a graph
\begin{definition}
A graph $G$ is a pair $(V, E)$, where $V$ is a set of vertices and $E \subseteq
V \times V$ is a set of edges. We say an element $i \in V$ is connected with
element $j \in V$ when $(i, j) \in E$. We will denote this by $ i \connected j$. We call $G$ a directed graph if we 
consider the elements of $E$ as ordered pairs. Otherwise, $G$ is said to be 
undirected. If all vertices have the same degree, then such a graph is called a regular graph.
\end{definition}

Next we introduce the concept of a hypergraph.
\begin{definition}
An undirected \emph{hypergraph} $H$ is a pair $(V ,E)$, where $V$ is a set of
vertices as in the traditional graph case and $E$ is the set of edges defined as
a collection of subsets of vertices $E \subseteq 2^V$. If for every $e \in E$ we
have $|e|=k$ we call the hypergraph $k$-regular.
\end{definition}
Note that any 2-regular hypergraph is an ordinary graph. An example of a
hypergraph is shown in Fig.~\ref{fig:hypergraph}
\begin{figure}[h]
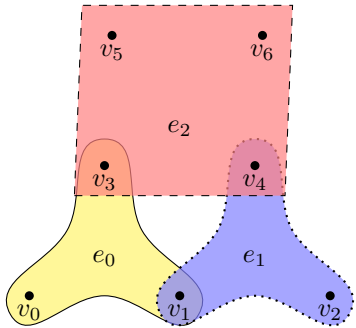

\centering\hyperpic \caption{A sketch of a hypergraph. Here $v_1,\ldots,v_6$
denote the vertices with two 3-hyperedges $e_0, e_1,$ and one 4-hyperedges
$e_2$.}\label{fig:hypergraph}
\end{figure}

\section{Walks}\label{sec:walks} In this section we present three
well-established models of quantum walks on graphs.

\subsection{The coined quantum walk}
Here we introduce the coined quantum walk model. We start with a simple walk on 
a line. Later we discuss the scattering walk model and move to arbitrary graphs.

In the simplest case a discrete time coined quantum walk on a line is given by a
bipartite system $\HH_C\otimes\HH_P$, where $\HH_C$ is a two-dimensional
Hilbert space with basis $\left\{ \ket{0} , \ket{1} \right\} $ and $\HH_P
=\SPAN(\left\{  \ket{n} : n \in \{0,\ldots,N-1\} \right\})$ is a position space.
Every step of the evolution $U$ is a composition of the coin and shift operators
\begin{equation}
\ket{\psi_{t+1}} = S(C\otimes\1_N) \ket{\psi_t},
\end{equation}
where $\ket{\psi_0}\in\mathcal{H_C}\otimes\mathcal{H_P}$ is some initial state.
The coin operator has a nice, tensor product because each vertex has the same 
degree $d=2$ which is equal to the dimensionality of the coin space. 

\subsection{Scattering quantum walk}

In order to model a coined quantum walk on a graph which is not regular, we must
modify this simple approach considerably. This can be achieved via the
scattering walk approach. In this case, we introduce separate coin operators
for every vertex $v$ of an graph $G$. Here, $G$ can be either directed or not.
Denoting the degree of the $i$\textsuperscript{th} vertex as $d_i$, we have that
each coin operator $C_i$ acts on $\C^{d_i}$. The entire space is
\begin{equation}
\XX = \C^{d_1} \oplus \ldots \oplus \C^{d_N}.
\end{equation}
The shift operator performs the scattering on a vertex $i$ given by formula 
\begin{equation}
S\ket{i,j}=\ket{j,i}.
\end{equation}
Let us consider a 
particle coming to vertex $j$ from some vertex $i$, \ie\ moving along the edge 
$(i,j)$. It becomes scattered after the shift operation, meaning that with 
equal probability it gets transfered to all other edges outgoing from $j$ and 
gets reflected back to $i$ along the $(j, i)$ edge, provided it exists in $G$. 
Hence we have for every vertex $i$ we have~\cite{hillery2003quantum}
\begin{equation}
U \ket{i, j} = r_i\ket{j, i}+t_i\sum_{ 
	v \connected j, v \neq i
	 } \ket{j,v}.
\end{equation}
Of course unitarity requires $|r_i|^2+(\deg(v_i)-1)|t_i|^2=1$ and the entire evolution is 
given by this formula.
 In the case of a directed graph $G$, we need to remember that each
edge $(i,j)$ can be seen as two directed edges.

For regular structures this simplifies to
\begin{equation}
C = C_0^{\oplus N}\cong C_0 \otimes \1_N.
\end{equation}
Let consider as an example the simple case when $r_i=1$ for every $i$. In this
case the shift operator for each edge subspace $\XX_{i,j}=\SPAN(\{\ket{i,j},
\ket{j,i}\})$ acts as $\NOT$ operator. Using a different basis, we may write the
space $\XX$ as $\XX_E=\bigoplus_{i\connected j}\XX_{i,j}$. Then, the shift
operator takes the form
\begin{equation}
S'=\NOT^{\oplus |E|} \cong \NOT \otimes \1_{|E|}.
\end{equation}
The shift operator is a block operator in 
\begin{equation}
\XX=\bigoplus_{i\connected j} \SPAN(\{\ket{i,j}, \ket{j,i}\})
\end{equation} and the coin operator is a block 
operator in 
\begin{equation}
\mathcal{Y}=\bigoplus_{j\in V} \SPAN(\{\ket{i,j}\}_{i\connected j}).
\end{equation}
Additionally we allow cases, when operator $C$ changes in time in a cyclic
manner and call such model a generalized coined walk model.

\subsection{The Szegedy walk model} The Szegedy walk model was first introduced
in~\cite{szegedy2004quantum} as a model which allows quantization of arbitrary
Markov chain based algorithms. The model is as follows. We start with undirected
graph $G_C=G(V,E)$ and we set a bipartite graph $G_S=G(V \cup V', F)$ where $V'$
is the same as $V$ with all elements primed. As for the edges we have $(i,j')
\in F$ 
if and only if $(i,j)\in E$. The evolution is given by the unitary operators
$U_1, U_2 U_1, U_1 U_2 U_1, \ldots$ acting on the space $\SPAN(\{\ket{x,y'}: x
\in V, y' \in V' \})$.

We define reflections
\begin{equation}
\begin{split}
U_1=2 \sum_V \proj{d_v} - \1, \\
U_2=2 \sum_V \proj{\overline{d}_v} - \1,
\end{split}
\end{equation}
and unit vectors
\begin{equation}
\begin{split}
\ket{d_v}=\ket{v} \otimes \sum_{w\connected v} a_{v,w} \ket{w},\\
\ket{\overline{d}_v}=\sum_{w\connected v} a_{v,w} \ket{w} \otimes \ket{v},
\end{split}
\end{equation}
where $a_{v,w}$ are complex constants.

Usually, the unitary operators driving the evolution for the Szegedy walk model 
are chosen as presented above. In our work, we
assume that unitary operators can be chosen arbitrary with only assumption of
respecting the graph structure \ie\ the movement between not connected vertices 
is forbidden. 

\subsection{The staggered walk}

To formally introduce the staggered walk model, we first introduce the following
definitions. We will follow the naming used by Portugal 
\etal~\cite{portugal2016staggered}
\begin{definition}
A \emph{tessellation} of a set $A$ is a collection $\alpha = \{ p_k\}_k$ 
of subsets of $A$, $p_k \subset A$, such that $\bigcup_k p_k = A$ and 
$p_k \cap p_{k'}=\emptyset$ for $k \neq k'$.
\end{definition}

\begin{definition}
A tessellation of a graph $G=(V,E)$ is a tessellation of $V$ such that 
each $p_k$ forms a clique or is a single vertex. We will call $p_k$ a polygon.
\end{definition}
Note that this definition allows for a polygon to contain a single vertex. The
staggered quantum walk on a graph is defined using at least one graph
tessellation.

\begin{definition}
Given a graph $G(V,E)$ and its $n$ tessellations $\alpha_1,\ldots,\alpha_n$, 
$\alpha_k=\{p_{k,i}\}_i$, for $k = 1, \ldots, n$, the
staggered quantum walk is defined by the evolution operator $U \in 
\Urm(\XX^V)$, where $\XX^V = \C^{|V|}$:
\begin{equation}
U = U_n \dots U_2 U_1,
\end{equation}
where
\begin{equation}
\begin{split}
U_k = &2\sum_{i=1}^{|\alpha_k|} \proj{d_{k,i}} - 
\idop[\XX^V].
\end{split}
\end{equation}
The states $\ket{d_{k,i}}$ are:
\begin{equation}
\begin{split}
\ket{d_{k,i}}= \sum_{j \in p_{k,i}} a_{k,j} \ket{j}
\end{split}
\end{equation}
where $a_{k,j}$ are complex amplitudes.

For the staggered quantum walk model we assume, that unitary operators can be 
also chosen arbitrarily.
\end{definition}

\section{Hyperwalk model}\label{sec:model} In this section we introduce the
concept of quantum walks on hypergraph networks along with some intuitions. We
will call this model the \emph{quantum hyperwalk} model.

Now we want to emphasize that the coined quantum
walk model can be described as a
composition of two operators that take block operator form with respect to two
different decompositions (tessellations) of the computational basis. The main
restriction in the model is that the decomposition (tessellation) corresponding
to the edges of the graph always consists of sets with two basis states. We aim
at loosening this restriction and developing a quantum walk model suitable for
hypergraphs, in this sense, that the tessellation of a 
hypergraph
$H=(V,E)$ is tessellation of basis states $\left\{ \ket{v,e}: v \in V,e \in E
\right\}$.

\begin{definition}
We define a hyperwalk on a hypergraph $(V, E)$ as a composition $U^E U^V$ 
of two
unitary operators: $U^V$ and $U^E$ on the space $\XX=\SPAN(\{\ket{v,e}: e \in 
E\wedge v \in e\})$, where
\begin{equation}
U^V = \sum_{v\in V} C_v,
\end{equation}
\begin{equation}
U^E = \sum_{e\in E} S_e,
\end{equation}
for $C_v$ being a coin operator for a fixed vertex $v$ acting on
$\SPAN (\{\ket{v,e} \}_{e \in E})$ and $S_e$ being a shift operator for a fixed 
edge $e$
acting on $\SPAN ( \{\ket{v,e} \}_{v \in e})$. 
\end{definition}

We also introduce a generalized version of this model. By a \emph{generalized
hyperwalk} we mean an instance of a hyperwalk for which the underlying unitaries
change with time.

Here we present an example of construction of such walk.

\begin{example}
We define $C_v=\1_{\deg(v)}-2\ketbra{\psi_v}{\psi_v}$ for
$\ket{\psi_v}=\frac{1}{\sqrt{\deg{(v)}}} \sum\limits_{\{e: v \in 
e\}}\ket{v,e}$ and 
$S_e=\1_{|e|}-2\ketbra{\psi_e}{\psi_e}$ 
for $\ket{\psi_e}=\frac{1}{\sqrt{|e|}}\sum_{v \in e}\ket{v,e}$ obtaining a 
hypergraph generalization of the
Grover's walk. Let us note that for a 2-regular hypergraph, \ie\ an ordinary
graph, we obtain $S_e$ which are two dimensional Grover's diffusion operator
that are equal to $\NOT$. This shows that our model recovers the proper
behaviour for a hypergraph which reduces to an ordinary graph.
\end{example}
\begin{example}
The idea of a hyperwalk gives the possibility to implement walks on directed
graphs. The basic way to ensure that  computation performed with use of directed
connections is reversible (unitary) is to ensure that for each vertex the number
of directed inputs and outputs is the same. In order to satisfy this condition
for a finite graph the directed connections must contain loops, which may be
seen as hyperedges. Thus, we define a quantum walk with directed edges as
\begin{equation}
S_E = \sum_{(v_1, ..., v_{l_e})=e\in E} \sum_{i=1}^{l_e} 
\ketbra{v_{i+1},e}{v_i,e},
\end{equation}
where $E$ is a set of edges defined as ordered sequences of vertices. In 
the case of 2-element edges we recover the canonical shift operator. For a
hyperedge we obtain cyclic shift among the loop constructed by this edge.
\end{example}

Additionally we allow the case when operators $U^V, U^E$ change in time in a
cyclic manner and call such model a \textit{generalized hyperwalk}.
 
Hyperwalk model can be seen as a natural generalization of the coined walk
model. This generality comes from two facts. First of all, by using hyperedges,
it is possible to construct higher dimensional space of basis states than in coined
walk. Second, loosening the restriction of shift operator to be a permutation
matrix, gives us additional dynamics in the constructed space.

On the other hand, sometimes, for defined graphs (hypergraphs), it is not 
trivial or it is not even possible to obtain the given walk evolution by using 
the hyperwalk model, despite the graph structures allows us to do so. The 
explanation of this problem and the formal comparison of introduced models is 
left to the next section.
\section{Relations between models}\label{sec:relation}

In this section we want to compare walk models
discussed in the previous sections. To do this we introduce two alternative
definitions of comparing walks and next we present our results of comparing
quantum walk models.

To clarify notation let $QW_A(G, \ket{\psi}, n)$ denote the state after $n$
iterations of discrete-time quantum walk model $A$ on graph $G$ with initial an
state $\ket{\psi}$. We also use the notation $\mathcal{P}$ for measurement on
vertices, where the probability of finding state $\ket{\psi}$ in vertex $v$ is 
denoted by $\mathcal{P}(\ket{\psi})(v)$.

\begin{definition}\label{def1}
	For given two models $A$ and $B$, we say that model $A$ is an instance of 
	model $B$ ($A \preceq B$) when for all graphs $G_A(V_1,E_1)$, measurements
	$\mathcal{P}_A$ and initial states $\ket{\psi_A}$ there exists 
	$G_B(V_2,E_2)$, where $V_1 \subset V_2$, measurement $\mathcal{P}_B $ and 
	initial state $ \ket{\psi_B}$, such that for all $n_0\in \mathbb{N}_0$, 
	exists $n_1 \in \mathbb{N}_0$ we have 
	\begin{equation}
	\mathcal{P}_A(QW_A(G_A,\ket{\psi_A},n_0))(v) = \mathcal{P}_B(QW_B(G_B,\ket{\psi_B},n_1))(v).
	\end{equation}
\end{definition}

Unfortunately, this definition allows us to find equivalence between walk 
models, which are very loosely based on the underlying structure. Some examples 
of equivalence of walk models which can be derived from this definition are 
presented below.

\begin{enumerate}
\item Szegedy $\preceq$ coined and coined $\succeq$ Szegedy: In the coined walk
model we are given a set of basis states $\ket{i,j}$, where $i,j \in V$ if
$(i,j) \in E$ and in Szegedy walk model our basis states are $\ket{i,j'}$ if
$(i,j) \in E$. There exists a bijection $\Xi$, that is defined as $\Xi
(\ket{i,j}) = \ket{i,j'}$. Now we can assume, that if $C$ is a given operator we
take $U_1=C$ and $U_2$ is constructed as in definition of Szegedy walk model. If
$U_1$ is given then we put $C=U_1$. Then, as we can observe, the equalities are
satisfied
\begin{equation}
\begin{split}
\Xi(SCSC\ket{\psi})=U_2 \Xi (SC \ket{\psi})=U_2 U_1 \Xi(\ket{\psi}).
\end{split}
\end{equation}
The above follows from the fact both space are the same up to labeling. This 
observation implies that measurements are connected by $P_{Sz}=\Xi(P_C)$, hence 
they are the same.

\item hyperwalk $\preceq$ coined: For a given hypergraph $H(V,E)$ with the
evolution operator $U_{H}=U^EU^V$, an initial state $\ket{\psi_0}_{H} = 
\ket{v,e}$ and a measurement $\mathcal{P}_{H}(\ket{\psi})(v)=\sum\limits_{e \in 
	E} |\braket{v,e}{\psi}|^2$, we create a bipartite graph $G(V_1 \cup 
V_2,F)$, where the sets $V_1=V$ and $V_2=E$ are the partitions of a bipartite graph, and $(v, e) \in F$ if and only if $v 
\in V_1$ is contained in an edge $e_1 \in E$. We construct a space with 
the basis vectors $\{\ket{v,e}, \ket{e,v} : v \in V, e \in E  \}$. We 
define the shift operator in standard form 
\begin{equation}
\begin{split}
S\ket{v,e}=\ket{e,v} \\
S\ket{e,v}=\ket{v,e},
\end{split}
\end{equation}
and the coin operator as
\begin{equation}
C=U^V \oplus U^E.
\end{equation}

Consequently, the evolution in quantum coined walk model with the initial state
$\ket{\psi_0}_{C}= \ket{v,e}$ is given by
\begin{equation}
\begin{split}
U_{C}=SCSC=SC(\ketbra{1}{0}\otimes U^V+\ketbra{0}{1}\otimes V^E)
\\
=S(\ketbra{1}{0}\otimes U^EU^V+\ketbra{0}{1}\otimes U^VU^E)= U^EU^V \oplus U^VU^E.
\end{split}
\end{equation}
Then the measurement should be 
$\mathcal{P}_{C}(\ket{\psi})(v)=\sum\limits_{e \in E} 
|\braket{v,e}{\psi}|^2$ and now it is clear that 
$\mathcal{P}_{H}(U_{H}^n\ket{\psi_0}_H)=\mathcal{P}_{C}(U_{C}^n\ket{\psi_0}_C)$.

\item generalized hyperwalk $\preceq$ staggered: We are given hypergraph
$H(V,E)$ with the evolution operator $U_{G_H,k}=U^E_kU^V_k$ ,for $k \in \{ 1,
\ldots, K \} $, an initial state $\ket{\psi_0}_{G_H} = \ket{v,e}$ and a
measurement $\mathcal{P}_{G_H}(\ket{\psi})(v)=\sum\limits_{e \in E}
|\braket{v,e}{\psi}|^2$. Let us note that we can consider an $N$ dimensional
system for staggered walk, where $N$ is the number of basis states in
generalized hyperwalk. We can set $W=\{\ket{v,e} : v \in V, e \in E, v \in e \}$
as a set of vertices for graph, which defined staggered walk on it and take
initial state $\ket{\psi_0}_S=\ket{\psi_0}_{G_H}$.  We introduce such
tessellations for which the unitary matrices $ U^{E}_i,U^{V}_j$ can be treated
as evolution operators. The measurement on the staggered walk model works on
proper group of vertices i.e.
$\mathcal{P}_{S}(\ket{\psi})(v)=\sum\limits_{\ket{v,e} \in W}
|\braket{v,e}{\psi}|^2$. It is clear that we obtain the same evolution, because
as introduced is this case base states in staggered walk are the same as in
generalized hyperwalk model.

\item staggered $\preceq$ generalized coined: We are given graph $G(V,E)$ with
$N$ vertices, $V=\{1,\ldots,N\}$, $k$ tessellations and unitaries $U_1, \ldots,
U_k$. We take an initial state $\ket{v_0}_S$ and measurement
$\mathcal{P}_{S}(\ket{\psi})(v)=|\braket{v}{\psi}|^2$. We construct a new graph
by adding to the set $V$ one vertex $t_{i,j}$ for each polygon $j$ in each
tessellation $i$. We connect every newly added vertex associated with some
polygon to vertices included in this polygon. Let us denote the basis states by
$\{\ket{v, t_{i}},\ket{t_{i},v}\}$, where $v \le N$, $i\le k-1$. Here, we
omitted the second index in vertices $t_{i,j}$, because parameters $v$ and $i$
are sufficient to determine vertex $t_i$ uniquely in state $\ket{v,t_i}$. The
generalized coined walk on this graph will be represented by $SCSC_{k-1}\ldots
SCSC_0$. Here, $S$ is the standard shift operator, $C$ is defined by
\begin{equation}
\begin{split}
C\ket{v,t_{i}}=\ket{v,t_{(i+1)_k}}, \\
C\ket{t_{i},v}=\ket{t_{i},v}.
\end{split}
\end{equation}
and $C_i$ is defined as 
\begin{equation}
\begin{split}
	C_i\ket{v,t_{j}}=\ket{v,t_{j}}, \\
	C_i\ket{t_{j},v}=\ket{t_{j},v}, \quad j \not= i \\
	C_i \ket{t_{i},v}= \sum_{w \in V} \bra{w}U_{i+1}\ket{v} \ket{t_{i},w}.
\end{split}
\end{equation}
The initial state can be chosen to be $\ket{\psi_0}_{G_C}= \ket{t_{0},v_0}$ and
the measurement of vertex $v$ is
$\mathcal{P}_{G_C}(\ket{\psi})(v)=\sum_{\ket{t_{i},v}}
|\braket{t_{i},v}{\psi}|^2$. To see the equality between both measurements, we
start with $\ket{v_0}_S$ and after the first step we obtain $U_1\ket{v_0}_S$ in
the staggered walk model. Assuming that the first step in the generalized coined
walk is given by $SCSC_0$, we get
\begin{equation}
\begin{split}
SCSC_0\ket{t_0,v_0}_{G_C}=SCS\sum_{w \in V}\bra{w}U_1\ket{v_0}_S\ket{t_0,w}=SC\sum_{w \in V}\bra{w}U_1\ket{v_0}_S\ket{w,t_0}
\\
=S\sum_{w \in V}\bra{w}U_1\ket{v_0}_S\ket{w,t_1}=\sum_{w \in V}\bra{w}U_1\ket{v_0}_S\ket{t_1,w}.
\end{split}
\end{equation}
Then $\bra{w}U_1\ket{v_0}_S=\bra{t_1,w}SCSC_0\ket{t_0,v_0}_{G_C}$, so finally we get
 $$\mathcal{P}_{G_C}(SCSC_{k-1} \ldots SCSC_0 \ket{t_0,v_0}_{G_C})(v)=\mathcal{P}_{S}( U_k \ldots U_1\ket{v_0}_S)(v).$$

\item generalized coined $\preceq$ coined: For a given graph $G$ with $N$
vertices,$\{1,\ldots,N\}$  with changing in time coins $C_0,\ldots,C_{k-1}$, an 
initial state $\ket{\psi_0}=\ket{v,w} $ and a measurement 
$\mathcal{P}_{G_C}(\ket{\psi_0})(v)=\sum_{w \in V} |\braket{v,w}{\psi_0}|^2$. We
introduce a new graph with $kN$ vertices $v^{(i)}$, where $v \le N$, $i \in
0,\ldots, k-1$. In this graph we have connections only between vertices
$v^{(i)}, w^{(i+1)_k}$, for $v ,w\le N, i \in 0,\ldots, k-1$ if and only if
$v\connected w$ in $G$. This generates new basis states
$\{\ket{v^{(i)},w^{(i+1)_k}},\ket{w^{(i+1)_k},v^{(i)}}\}$. In this model the
initial state will be $\ket{v^{(0)},w^{(k-1)}}$ and the state will evolve to
vertices with higher indexes and eventually come back to vertices with the index
zero. This means, we define the coin operator as
\begin{equation}
\begin{split}
C\ket{v^{(i)},w^{(i-1)_k}}=\sum_{z \in V} \bra{v,z}C_i\ket{v,w}\ket{v^{(i)},z^{(i+1)_k}},
\\
C\ket{v^{(i)}, w^{(i+1)_k}} = \ket{v^{(i)}, w^{(i-1)_k}}.
\end{split}
\end{equation}
The measurement on the
vertex $v$ is given on states associated with $v^{(i)}$ i.e.
$\mathcal{P}_{C}(\ket{\psi_0})(v)= \sum\limits_i \sum\limits_{w \connected v}
|\braket{v^{(i)},w^{(i-1)_k}}{\psi_0}|^2$.

After the first iteration in the coined walk model we have
\begin{equation}
\begin{split}
SC\ket{\psi_0}=SC\ket{v^{(0)},w^{(k-1)}}=S\sum_{z \in V} \bra{v,z}C_0\ket{v,w}\ket{v^{(0)},z^{(1)}}\\=\sum_{z \in V} \bra{v,z}C_0\ket{v,w}\ket{z^{(1)},v^{(0)}}.
\end{split}
\end{equation}
On the other hand, considering the generalized coined walk model gives us
\begin{equation}
\begin{split}
SC_0\ket{\psi_0}=SC_0\ket{v,w}=\sum_{z \in V} \bra{v,z}C_0\ket{v,w}\ket{z,v}.
\end{split}
\end{equation}
We can observe, that the both models give us the same evolution, which implies
that the measurement outcomes will be exactly the same.
\end{enumerate} 
	
As it can be seen, according to this definition the models are equivalent. This
result should not be surprising, as we are allowed to compare models $A$ and $B$
on graphs with different structures. For the case, when $A \preceq B$, the
quantum walk model $B$ does not have to express the idea of random walk on graph
$G_A$.
\begin{figure}[h]
\centering
\includegraphics[height=3cm, width=4cm]{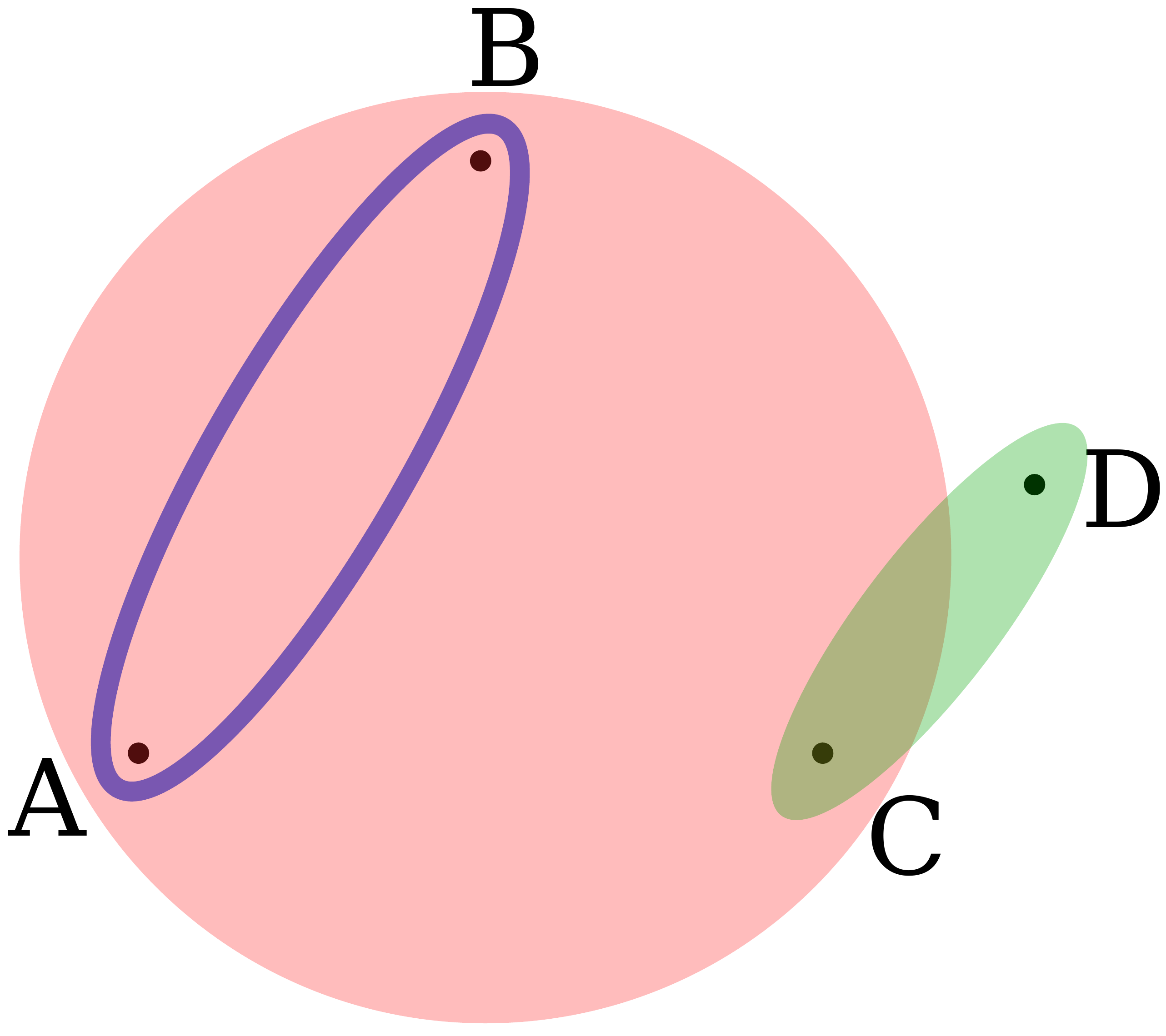}
\caption{Example of hyperwalk. We have the following 3 hyperedges:
	$a=\{A,B,C\},b=\{A,B\},c=\{C,D\}.$}
\label{fig:graph-example3}
\end{figure}
For example if we want to change generalized hyperwalk with 2 different
distributions on the graph shown in Figure~\ref{fig:graph-example3} into the
Szegedy walk, we need to take graph with $336$ vertices. According to the
previous discussion changing the generalized hyperwalk to a staggered walk costs
$7$ vertices. We see that if we want to put this model to generalized coined
walk we should take a graph with $21$ nodes and $8$ coin operators. In the next
step, it is necessary to model coined walk on a graph with $8 \times 21 = 168$
vertices. The last step is cloning of the vertices to obtain Szegedy walk, so we
end up with $336$ nodes required. Of course there still can exist methods to
achieve this result with a smaller number of vertices, but this example is
introduced to show problems which can appear. That is why we introduce a new
concept of comparing two quantum models.

\begin{definition}\label{def2}
Let $A$ and $B$ be two models of a quantum walk. We say that model $A$ is
strongly an instance of model $B$ ($A \prec B$) when for all graphs $G_A(V,E)$
there exists a graph $G_B$, such that the number of basis states in model $B$ is
no greater then the number of basis states in model $A$. Moreover, for all
initial states $\ket{\psi_A}$ there exists an initial state $ \ket{\psi_B}$,
such that for all $n\in \mathbb{N}_0$ and $v\in V$, we have
\begin{equation}
\mathcal{P}_A(QW_A(G_A,\ket{\psi_A},n))(v) = 
\mathcal{P}_B(QW_B(G_B,\ket{\psi_B},n))(v).
\end{equation}
\end{definition}
Based on this definition, we show that every staggered walk is an instance of a 
generalized hyperwalk. Furthermore, we do not need to deeply change the 
structure of the initial graph in order to obtain this behavior.
\begin{theorem}
According to Definition \ref{def2}every staggered walk is an instance of 
the generalized hyperwalk.
\end{theorem}
\begin{proof}

For a given graph with the staggered walk, defined by unitary matrices
$U_1,\ldots,U_n$, we introduce a hypergraph with the same number of vertices and
one hyperedge containing all vertices suitable for generalized hyperwalk. One
can see that the spaces for both walks have the same dimensionality hence there
exists a bijection between spaces on staggered and generalized coined walk
models. So we can assume that the coin operator is constant and it is given by
the identity matrix. The shift operator is changing in time in the same manner
as the unitary operators $U_1,\ldots,U_n$ for each tessellations, namely
$U^E_k:=U_k$. If the measurement for staggered walk is
$\mathcal{P}_S(\ket{\psi})(v)=|\braket{v}{\psi}|^2$, then we take
$\mathcal{P}_{GH}= |\braket{v,e}{\psi}|^2$.
\end{proof}

\section{Conclusions}\label{sec:conc}

In this work we introduced a model of quantum walks on hypergraphs and a
generalized version of this model. By generalized we mean that the evolution
operators associated with the walk might change in time. We introduced two
non-equivalent definitions of the case when one quantum walk model is an
instance of another model. The first definition of this equivalence allows us to
heavily manipulate the underlying graph structure of the walk. Using this
definition we shown that hyperwalk model is equivalent to a coin model and the
same for the generalized version.
 
Next, we introduced a stronger version of the equivalence of walk models. In it,
we enforce the graph to be a minimal graph necessary for a given model. In this
regime we were able to show that a generalized hyperwalk introduces in fact new
dynamics. This result completes Table~\ref{tab:motiv} and shows that a quantum
walk on a hypergraph is a generalization of the staggered walk model.

\section*{Acknowledgments}
This work was supported by the polish National Science Centre under project 
numbers 2015/17/B/ST6/01872	({\L}P and PL) and 2016/22/E/ST6/00062	(RK).

\bibliographystyle{ieeetr}
\bibliography{../hyperwalks}
\end{document}